\documentclass[letterpaper, 10 pt, conference]{ieeeconf}  

\IEEEoverridecommandlockouts                              
\usepackage{graphics, graphicx} 
\usepackage{epsfig} 
\usepackage{mathptmx} 
\usepackage{times} 
\usepackage{amsmath, amscd, amsfonts} 
\usepackage{amssymb}  

\usepackage[colorlinks=true, linkcolor=black, urlcolor=blue, citecolor=black]{hyperref}
\usepackage{cleveref}
\usepackage{xcolor,color}
\usepackage{tikz}
\usepackage{pgfplots}

\usepackage{placeins}    
\usepackage{comment}

\newtheorem{theorem}{Theorem}[section]

\newtheorem{lemma}[theorem]{Lemma}

\newtheorem{remark}[theorem]{Remark}
\newtheorem{assumption}[theorem]{Assumption}

\newenvironment{customproof}[1][Proof]
{\par\noindent\textit{#1}\ignorespaces}
{\hfill$\blacksquare$\par}

\pgfplotsset{compat=1.18}
\begin{document}
	
\title{\bf Predictor-Feedback Stabilization of Linear Switched Systems with State-Dependent Switching and Input Delay$^*$}
	
\author{
Andreas Katsanikakis$^{1}$, 
Nikolaos Bekiaris-Liberis$^1$ and 
Delphine Bresch-Pietri$^2$
\thanks{$^*$Funded by the European Union (ERC, C-NORA, 101088147). Views and opinions 
expressed are however those of the authors only and do not necessarily reflect those of the 
European Union or the European Research Council Executive Agency. Neither the European 
Union nor the granting authority can be held responsible for them.}
\thanks{$^1$The authors are with the Department of Electrical and Computer Engineering, 
Technical University of Crete, Chania, 73100, Greece. Emails: akatsanikakis@tuc.gr and nlimperis@tuc.gr.}
\thanks{
$^2$The author is with MINES ParisTech, PSL Research University CAS-Centre Automatique et Systemes, 60 Boulevard Saint Michel 75006, Paris, France. Email: delphine.bresch-pietri@minesparis.psl.eu.
}
}
\maketitle
\setlength{\parskip}{0pt}
\begin{abstract}
We develop a predictor-feedback control design for a class of linear systems with state-dependent switching. The main ingredient of our design is a novel construction of an exact predictor state. Such a construction is possible as for a given, state-dependent switching rule, an implementable formula for the predictor state can be derived in a way analogous to the case of nonlinear systems with input delay. We establish uniform exponential stability of the corresponding closed-loop system via a novel construction of multiple Lyapunov functionals, relying on a backstepping transformation that we introduce. We validate our design in simulation considering a switching rule motivated by communication networks. 
\end{abstract}

\setcounter{footnote}{1}
    
\section{Introduction} \label{sec:intro}
Systems with state-dependent switching and input delay appear in various applications. Among different potential application examples we discuss the following three. In {vehicle control}, state-dependent switching occurs between throttle and braking operations depending on speed/acceleration states, while input delays appear due to engine and other actuator dynamics, see, e.g., \cite{IOANNOY}, \cite{ACC}. When modeling epidemic spreading dynamics, state-dependent switching may appear due to external measures imposition, such as quarantine measures, implemented depending on, e.g., the number of infected individuals, while at the same time, input delays may arise due to lags between implementation of policymaker strategies and their effect on epidemic spreading, see, e.g., \cite{SIS_M}, \cite{SIR_D}, \cite{Stability_SIS}, \cite{CompartEpidemi}. In networked control systems, state-dependent switching originates from communication protocols and scheduling rules, while delays appear due to, e.g., communication constraints \cite{NetworkCom}, \cite{liu2015networked}. These examples highlight the fact that simultaneous presence of state-dependent switching and input delay is common in practice, motivating the need for development of control design methods for this class of systems. 

The majority of related existing results address either input delays in systems with time-dependent switching, or state delays in systems with state-dependent switching. For the former case, there exist methods aiming at input delay compensation relying, for example, on construction of Linear Matrix Inequalities (LMIs) \cite{FreqControlPowerSys}, \cite{Lin_discrete}, \cite{Xi_Corr}, \cite{Kp}, \cite{Switch_Communication}, \cite{B}, \cite{AHMED}, \cite{LIMA} and Lyapunov-Krasovski functionals \cite{Yue_Kras}, \cite{Mazenc_LKF}, or on the truncated predictor method \cite{Wu_Truncated}, \cite{Sakthi_Truncated}, \cite{C}. However, in these works the delay length or the system dynamics are typically restricted. Systems with state delays and state-dependent switching, which, in general, require different treatment than the case of input delay, are addressed in \cite{SCL_2023}, \cite{Aut2025}. Moreover, input delay of arbitrary length is compensated developing predictor-based designs in \cite{Chattering_DBP}, \cite{KONG_Stochastic}, which address systems with switched delays rather than switched dynamics, and in \cite{my_ECC}, \cite{my_TDS} that consider time-dependent switching. Predictor-based control of discrete-time systems when the switching signal itself is the control input are addressed in \cite{Seuret}, which, although related, it is a different problem. To the best of our knowledge, there exists no result addressing the problem of (long) input delay compensation for linear switched systems with state-dependent switching via predictor feedback.

In this paper we develop a new predictor-feedback control law for a class of linear switched systems with input delay, in which the switching signal depends on the state. For a given switching rule, the main element of our predictor-feedback control design is a novel construction of an exact predictor state, which enables complete delay compensation. Differently from our previous works \cite{my_ECC}, \cite{my_TDS}, here an exact predictor state construction is possible because, as the switching rule depends on the state, an implicit, implementable formula for the predictor state can be derived viewing the given switched system as a nonlinear system. We thus provide such a predictor state construction and we furthermore derive an alternative, semi-explicit formula for its computation, capitalizing on the fact that, at each mode, the system dynamics are linear. We establish uniform exponential stability of the closed-loop system via construction of multiple Lyapunov functionals relying on a backstepping transformation that we introduce. We validate the control design’s performance in simulation using a two-mode system with a switching rule whose form is motivated by communication network applications.

The remainder of the paper is organized as follows. Section~\ref{sec2} presents the class of switched systems with input delay considered and the predictor-feedback control design. Section~\ref{sec: stab analysis} provides the stability analysis, where the main result is presented. Section~\ref{sec: sims} illustrates the theoretical results through a simulation example. Finally, Section~\ref{sec5} provides concluding remarks.

\section{Problem Formulation and Control Design} \label{sec2}
\subsection{Switched Linear Systems With Input Delay and State-Dependent Switching}  \label{sec2A}

We consider the following linear switched system with a constant delay in the input
\begin{equation}\label{1.1}
\dot{X}(t) = A_{\sigma(X(t))} X(t) + B_{\sigma(X(t))} U(t-D), 
\end{equation}
where $X\in \mathbb{R}^n$ is the state, $U \in \mathbb{R}$ is the control input, $D>0$ is an arbitrarily long delay, and $t \geq0$ is the time variable. The switching signal $\sigma : \mathbb{R}^n \to \mathcal{P}$ maps the state to a finite set of modes $\mathcal{P} = \{1, 2, \dots, p\}$. The state space is partitioned into a finite collection of disjoint regions $\{ \Omega_1, \Omega_2, \dots, \Omega_p \}$ such that
$\bigcup_{i=1}^p \Omega_i \;  = \mathbb{R}^n , \ \Omega_i \cap \Omega_j = \varnothing, \; \forall i \neq j,$
where each region $\Omega_i$ corresponds to a distinct mode of the switched system, i.e., $
    \sigma(X) = i \ \text{ if and only if } \ X \in \Omega_i$. Furthermore, we define $\Omega_{j,i} = \Omega_{i,j} \subseteq \Omega_i \cup  \Omega_j$ as the switching surfaces between adjacent regions.
    
We impose the following assumptions on the delay-free system. The first guarantees well-posedness of both the open- and closed-loop systems and the second implies that the delay-free system can be stabilized by a nominal control law of the form $U=K_{\sigma}X$.

\begin{assumption}\label{assumption1}
For any initial condition $X_0\in\mathbb{R}^n$ and any disturbance
$d \in L^2_{\mathrm{loc}}([0,\infty);\mathbb{R})$, the systems
\begin{align}
  \dot X(t) &= \big(A_{\sigma(X(t))}+B_{\sigma(X(t))}K_{\sigma(X(t))}\big)X(t), \label{open control loop} \\
  \dot X(t) &= A_{\sigma(X(t))}X(t) + B_{\sigma(X(t))}d(t), \label{open loop dist}
\end{align}
have a unique, absolutely continuous solution, in Carath\'eodory sense\footnote{See, e.g., \cite{contraction} and references therein, for specific conditions that guarantee this.} defined on $\mathbb{R}_+$. Moreover, the corresponding switching signal $\sigma(X(t))$ is piecewise constant and exists for all non-negative times (thus it does not exhibit Zeno behavior).
\end{assumption}

\begin{assumption}\label{delay-free as}
There exist families of vectors $K_i^{\rm T}$, of symmetric matrices $P_i$, $Q_i$, such that $Q_i>0$, and of positive constants $\alpha_i, \beta_i$,  such that the following hold

    \begin{equation}
        \alpha_i |X|^2 \le X^\top P_i X \le \beta_i |X|^2,\qquad \forall X\in\Omega_i,
    \end{equation}
    \begin{align}\label{mean_delayf_stabilityy}
        &X^T \left[\left({A_i} + {B_i} {K_i}\right)^T P_i + P_i \left({A_i} + {B_i} {K_i}\right) + Q_i \right] X \leq 0, \notag \\ &\forall  X \in \Omega_i, 
    \end{align}
    \begin{equation}\label{eq:exact_quadratic_switching}
            X^T (P_j - P_i) X = 0, \quad \forall X \in \Omega_{ij},
    \end{equation}
for all $i,j \in \mathcal{P}$.
\end{assumption}

\begin{remark}\label{remark}
Assumption~\ref{assumption1} guarantees well-posedness of both the open- and closed-loop systems, when employing predictor-feedback, under the given state-dependent switching law $\sigma(X)$. It also excludes the possibility of the appearance of Zeno behavior. 
Note that, for the considered partition of the state space, sliding phenomena are also excluded, while chattering phenomena are not. Such phenomena are undesirable for control implementation purposes, but they can be practically avoided via employing hysteresis switching (see e.g., \cite{Liberzon}, \cite{Antsaklis}).
Moreover, Assumption~\ref{delay-free as} guarantees that one can construct multiple, quadratic Lyapunov functions to study stability of the nominal (in the delay-free case) closed-loop system (\ref{open control loop}), as in, e.g., \cite{Liberzon}, \cite{Antsaklis}. However, condition (\ref{eq:exact_quadratic_switching}) is more restrictive as it imposes continuity at the switching surfaces of the respective Lyapunov functions. This is consistent with our setup in which $\sigma$ depends only on $X$ and where the $\Omega_i $ sets are such that $\Omega_{i,j}=\Omega_{j,i}$, while belonging to some $\Omega_i \,(\text{resp. }\Omega_j)$. Both assumptions may be restrictive, nevertheless, they are necessary here, in order to properly define the predictor state and to carry out the stability analysis, without the need to employ a hybrid systems framework.


\end{remark}

\subsection{Predictor-Feedback Control Design And Computation}

  The values of the predictor are the future values of the system state, and since the switching logic is state-dependent, we face a situation analogous to a non-linear system with input delay, where the predictor state formula is implicit, see, for example, \cite{MirkinKrsticBook}, \cite{MirkinKrstic2010}, and \cite{KarafyllisKrstic2017}. The reason for this is that the system parameters ($A_\sigma$, $B_\sigma$) depend on the switching mode, which in turn depends on the future state. Hence, the predictor state satisfies an implicit integral equation given by
\begin{equation}\label{P_theta}
    P(\theta) = X(t) + \int_{t-D}^{\theta} \left[ A_{\sigma(P(s))} P(s) + B_{\sigma(P(s))} U(s) \right] ds,
\end{equation}
for $t-D \leq \theta \leq t$.
Note that this definition is implicit because the switching mode \(\sigma(P(\theta))\) depends explicitly on the unknown future predictor state \(P(\theta)\). 
Having defined the predictor state \(P(t)\), we now propose the predictor-feedback control law
\begin{equation}\label{controller}
    U(t) = K_{\sigma(P(t))} P(t).
\end{equation}

Even though equation (\ref{P_theta}) defines the predictor implicitly, we can still obtain a more explicit formula for implementation, capitalizing on the fact that, at each given mode, the system is linear. 
To see this, since $\sigma$ is {state dependent}, the next switching instant is the first time the predictor trajectory leaves the current region. Hence, setting $m_1=\sigma(X(t))$ and defining recursively
\begin{equation}
    s_i \;=\; \inf\{\, s>s_{i-1}:\; P(t-D+s)\notin \Omega_{m_i}\,\},\quad i=1,2,\dots,k,
\end{equation}
where $k \in  \mathbb{N}_0$ is the number of switching instances within interval $[t,t+D)$ and  $m_i \in \mathcal{P}$ is the active mode on $[t+s_{i-1},t+s_i)$, by Assumption~\ref{assumption1}, this process leads to a finite switchings sequence 
\begin{equation}\label{s sequence}
    0=s_0 < s_1 < \cdots < s_k < s_{k+1}=D ,
\end{equation}
with $\sigma(P(\theta))\equiv m_i$ on
\begin{equation}\label{theta subs}
    t-D+s_{i-1} \le \theta < t-D+s_i .
\end{equation}
Hence, for each subinterval (\ref{theta subs}), we can explicitly solve \eqref{P_theta} as 
    \begin{align}\label{P_theta1}
        P(\theta) &= e^{{A_{m_i}}(\theta-t+D-s_{i-1})} X(t+s_{i-1}) \notag \\
        &\quad + \int_{t-D+s_{i-1}}^{\theta} e^{{A_{m_i}}(\theta - s)} {B_{m_i}}U(s) d s.
    \end{align}
Calculating \eqref{P_theta1} at $\theta=t-D+s_i$ gives 
\begin{align}\label{eq:segment}
X(t+s_i)&=e^{A_{m_i}(s_i-s_{i-1})}X(t+s_{i-1}) \notag \\
&\quad+\int_{t-D+s_{i-1}}^{t-D+s_i}e^{A_{m_i}(t-D+s_i-s)}B_{m_i}U(s)\,ds. 
\end{align}
Setting $i=k+1$ gives $X(t+s_{k+1})=X(t+D)=P(t)$. Hence, from (\ref{eq:segment}) we get the exact predictor formula at $t$, as
\begin{align}\label{P(t)}
    P(t) &= \prod_{n=1}^{k+1} e^{A_{m_n}(s_n-s_{n-1})}X(t) + \sum_{n=1}^{k+1} \left( \prod_{j=n}^{k} e^{A_{m_{j+1}}(s_{j+1}-s_{j})}  \right. \notag \\
         &\qquad \left. \times \int_{t-D+s_{n-1}}^{t-D+s_n} e^{A_{m_n}(t-D+s_n-\theta)} B_{m_n} U(\theta) d\theta \right).
\end{align}  

Although (\ref{P(t)}) provides the exact predictor expression, its implementation in each $t$ requires stepwise computations for each $\theta$-interval defined in (\ref{theta subs}) and detection of mode transitions. Hence, the computation of $P$ is in fact semi-explicit rather than fully explicit. In a practical implementation of (\ref{P(t)}), one could proceed in the following, algorithmic manner. Start from the available measurement $P(t-D)=X(t)$ in mode $m_1$, 
compute $P(\theta)$ via \eqref{P_theta1} until the trajectory leaves $\Omega_{m_1}$ at 
$\theta = t-D+s_1$ (this event can be detected at current time $t$ using the computed values of $P(\theta)$, which depend on $X(t)$ and $U(s), \,s\in[t-D,t]$), then switch to $m_2$ and repeat. If no exit occurs, then $k=0$ and 
\eqref{P(t)} becomes a single–mode expression. Once the process terminates, the 
sequence \eqref{s sequence} is fully determined and the semi-explicit formula 
\eqref{P(t)} can be implemented directly. This scheme may be computationally
advantageous when a minimum dwell time can be determined a priori,
depending on the shape of the $\Omega_i$ sets and the linearity of the
dynamics, allowing to derive a growth bound. This minimum dwell time 
may allow to rely on (\ref{P_theta1}) to fasten intermediate computations, without 
the need of checking for mode transitions at all sampling times.
\section{Stability Analysis}\label{sec: stab analysis}
Having developed the predictor-feedback control law, we are now ready to present our main result.
\begin{theorem}\label{Exponential Stability}
Consider the closed-loop system (\ref{1.1}) with the controller (\ref{controller}). Under Assumptions~\ref{assumption1} and~\ref{delay-free as}, for all $X_0\in\mathbb{R}^n$, $U_0\in L^2[-D,0]$ there exist positive constants $\rho,\xi$ such that the following holds
\begin{align}\label{stability equation}
\left| X(t) \right| + \sqrt{\int_{t-D}^{t} U(\theta)^2 d\theta} &\leq \rho \left( \left| X(0) \right| + \sqrt{\int_{-D}^{0} U(\theta)^2 d\theta} \right) \notag \\& \qquad \times e^{-\xi t}, \quad t \geq 0. 
\end{align}
\end{theorem}
The proof of Theorem \ref{Exponential Stability} relies on some lemmas, which are presented next, together with their proofs.

\begin{lemma}[{backstepping transformation}]\label{backstepping transformation}
    The following backstepping transformation, 
    \begin{equation}\label{W_theta}
        W(\theta)=U(\theta) - K_{\sigma{(P(\theta))}} P(\theta), \quad t-D \leq \theta \leq t,
    \end{equation}
    where $P(\theta)$ is obtained from (\ref{P_theta}) for $t-D \leq \theta \leq t$, together with the control law (\ref{controller}), transform system (\ref{1.1}) to the target system 
    \begin{align}
        \dot{X}(t) &= \left(A_{\sigma(X(t))}+B_{\sigma(X(t))}K_{\sigma(X(t))}\right) X(t) \notag \\ &\quad + B_{\sigma(X(t))} W(t-D) \label{Xd_trans} , \\
              W(t) &= 0,\quad t\geq0. \label{W_trans}
    \end{align}

    \begin{proof}
        System (\ref{1.1}) can be written as
        \begin{align}
            \dot{X}(t) &= \left(A_{\sigma(X(t))} + B_{\sigma(X(t))}K_{\sigma(X(t))}\right) X(t) \notag \\
                       &\quad + B_{\sigma(X(t))} \left(U(t-D) -  K_{\sigma(X(t))}X(t)\right). \label{2.16}
        \end{align}
        We now use (\ref{W_theta}). Setting $\theta=t-D$, from (\ref{P_theta}) we get $P(t-D)=X(t)$. Observing (\ref{P_theta}) and (\ref{2.16}), transformation (\ref{W_theta}) maps the closed-loop system consisting of the plant (\ref{1.1}) and the control law (\ref{controller}), to the target system (\ref{Xd_trans}), (\ref{W_trans}).
    \end{proof}
\end{lemma}

\begin{lemma}[{inverse backstepping transformation}]\label{inverse_transformation} 
The inverse backstepping transformation of $W$ is
\begin{equation}\label{inverse_theta}
    U(\theta)=W(\theta) + K_{\sigma{(\Pi(\theta))}} \Pi(\theta),
\end{equation}
where for $t-D \leq \theta \leq t$,
\begin{align}\label{Pi_theta}
    \Pi(\theta) &= X(t) + \int_{t-D}^{\theta}  \left[ \left( A_{\sigma(\Pi(s))}+B_{\sigma(\Pi(s))}K_{\sigma(\Pi(s))} \right) \Pi(s) \right. \notag   \\ & \left. \quad + B_{\sigma(\Pi(s))} W(s) \right] ds .
\end{align}

\begin{proof}
    Let $\theta \in [t-D,t]$. We observe from (\ref{W_theta}) that $U(\theta)=W(\theta) + K_{\sigma{(P(\theta))}} P(\theta)$. Solving the ODE (\ref{Xd_trans}) in a similar implicit way as for the original system it can be shown that $\Pi(\theta)=X(\theta+D)$, where $\Pi(\theta)$ is given from (\ref{Pi_theta}), and it holds that $\Pi(\theta)=P(\theta)$. 
\end{proof}

\end{lemma}

\begin{lemma}[{norm equivalency}]  \label{lemma_u(theta)}  
    For the direct transformation (\ref{W_theta}), the following inequality holds for some positive constant $\nu_1$ 
    \begin{equation}\label{ut_bound}
      |X(t)|^2+ \int_{t-D}^{t}{| W(\theta) |^2 d \theta} \leq \nu_1 \left( |X(t)|^2 + \int_{t-D}^{t}{| U(\theta) |^2 d \theta} \right).
    \end{equation}
    Similarly, for the inverse transformation (\ref{inverse_theta}), it holds for some positive constant $\nu_2$
    \begin{equation}\label{wt_bound}
      |X(t)|^2+ \int_{t-D}^{t}{| U(\theta) |^2 d \theta} \leq \nu_2 \left( |X(t)|^2 + \int_{t-D}^{t}{| W(\theta) |^2 d \theta} \right).
    \end{equation}

\begin{proof}
        From (\ref{W_theta}), for the direct transformation we apply Young's inequality to obtain 
        \begin{align}\label{W_theta_bound}
            \int_{t-D}^{t}{| W(\theta) |^2 d \theta} &\leq 2 \left( \ \int_{t-D}^{t}{ |U(\theta) |^2 d\theta} \right. \notag \\ 
                        &\left. \quad + {M_K}^2  \int_{t-D}^{t}{  \left| { P(\theta)   } \right| ^2 d \theta}  \right),
        \end{align}
        where for any matrix $R_i$, with $R$ being $A,B,K, H,$ we set
        \begin{equation}
         M_R = \max \{|R_0|,|R_1|,\ldots,|R_p|\}, \label{M_A}
        \end{equation}
        with $H_i=A_i + B_i K_i$, $i=1,\ldots,p$. Using the triangle inequality and (\ref{M_A}) in (\ref{P_theta}), it holds that 
        \begin{align}\label{bound_Pth_1}
          \left| P(\theta) \right| \leq \left| X(t) \right| + M_B \int_{t-D}^{t}|U(s)|ds + M_A \int_{t-D}^{\theta}|P(s)|ds.
        \end{align}   
        We can now apply Gronwall's inequality in (\ref{bound_Pth_1}) to obtain 
        \begin{align}\label{bound_Pth_2}
          \left| P(\theta) \right| \leq \left( \left| X(t) \right| + M_B \int_{t-D}^{t}|U(s)ds \right) e^{M_A(\theta-t+D)}.
        \end{align}  
        Applying Young's and Cauchy-Schwartz's inequalities we get from (\ref{bound_Pth_2})
        \begin{align}\label{bound_Pth_4}
          \int_{t-D}^{t}\left| P(\theta) \right|^2 d\theta  &\leq 2 e^{2M_AD} D \notag \\
          &\ \times  \left( \left| X(t) \right|^2 + M_B^2 D \int_{t-D}^{t}|U(\theta)|^2d\theta \right). 
        \end{align} 
        Combining (\ref{W_theta_bound}) and (\ref{bound_Pth_4}) we reach (\ref{ut_bound}),  where 
        \begin{align}
            \nu_{1} &=  \max \left\{ 4 M_{K}^2 D e^{2 M_A D}+1, 4 M_{K}^2 D^2 e^{2 M_A D} M_{B}^2 +2 \right\}.
        \end{align}  
        Analogously, using the inverse transformation from (\ref{inverse_theta}), we can similarly prove (\ref{wt_bound}) via (\ref{Pi_theta}), where
        \begin{align}
            \nu_{2} &=  \max \left\{ 4 M_{K}^2 D e^{2 M_H D}+1, 4 M_{K}^2 D^2 e^{2 M_H D} M_{B}^2 +2 \right\}.
        \end{align}
    \end{proof}
\end{lemma}

\begin{lemma}[{stability of the target system}]\label{trans_stability2}
    Under As- \\ sumptions \ref{assumption1} and \ref{delay-free as}, there exist positive constants $\kappa$, $\mu$, such that for the target system (\ref{Xd_trans}), (\ref{W_trans}), the following holds
    \begin{align}\label{trans_stability_result2}
        \left| X(t) \right|^2 +  \int_{t-D}^{t} W(\theta)^2 d\theta &\leq \kappa \left( |X(0)|^2 + \int_{-D}^{0} W(\theta)^2 d\theta \right)   \notag \\ & \ \times e^{-\mu t}, \quad t \geq 0.
    \end{align}

    \begin{proof}
         According to Lemma \ref{backstepping transformation}, for each subsystem of the family described by the target switched system (\ref{Xd_trans}), (\ref{W_trans}), i.e., for any $i \in \mathcal{P}$, it holds
        \begin{align}
            \dot{X}(t) &= \left(A_{i}+B_{i}K_{i}\right) X(t) + B_{i} W(t-D) \label{Xdp_trans} , \\
                  W(t) &= 0. \label{Wp_trans}
        \end{align} 
        Consider an arbitrary time window $[t_0, t_1)$ such that $X(t) \in \Omega_i$ and $\sigma(X(t)) = i$ for all $t \in [t_0, t_1)$ (such a time interval exists by Remark~\ref{remark} since (\ref{Xdp_trans}), (\ref{Wp_trans}) is equivalent to (\ref{open control loop}) for $t \geq D$, while for $t<D$ it is equivalent to (\ref{open loop dist}) with $d(t)=U_0(t-D)$). Within this interval, the system (\ref{Xdp_trans}), (\ref{Wp_trans}) operates without switchings. Therefore, we can assign to any time $t \in [t_0,t_1)$ the following Lyapunov functional 
    \begin{equation}\label{LyapunovFnc2}
        V_i(t) = X(t)^T P_i X(t) + b \int_{t-D}^{t} e^{(\theta+D-t)}W(\theta)^2 d\theta.
    \end{equation}
     Calculating the derivative of (\ref{LyapunovFnc2}), along the solutions of the target subsystem for $t \in [t_0 , t_1)$, we obtain
    \begin{align}
        \dot{V}_{{i}}(t) &\leq -X(t)^T  Q_{i}  X(t) + B_{i}^T W(t-D)P_{i} X(t) \notag \\ \
        &\quad  + X(t)^T P_{i} B_{i} W(t-D) + b e^D W(t)^2 \notag \\
        &\quad - b W(t-D)^2 - b\int_{t-D}^{t} e^{(\theta+D-t)}W(\theta)^2 d\theta. \label{2dv_t_1}
    \end{align}
    Observing $-X^T Q_{i}  X\leq - \lambda_{\min}({Q_{i}}) |X|^2$, applying \\ Young's inequality, and choosing
    \begin{equation}
        b= \max \limits_{i=1,\ldots,p}\left\{ \frac{ 2 |B_{i}P_{i}|^2 }{ \lambda_{\min}({Q_{i}}) } \right\}, 
    \end{equation}
    we get from (\ref{2dv_t_1}) that
    \begin{align}
    \label{2dv_t_11}
        \dot{V}_{i}(t) &\leq -\frac{1}{2} \lambda_{\min}({Q_{i}}) |X(t)|^2  - b\int_{t-D}^{t} e^{(\theta+D-t)}W(\theta)^2 d\theta.
    \end{align}
    Therefore, we conclude that 
    \begin{equation}\label{final derivative of V2}
        \dot{V}_{i}(t) \leq {-\mu_i }V_{i}(t),\quad t \in [t_0,t_1),
    \end{equation}
    where 
    \begin{align}\label{alpha2}
        \mu_{i} &= \min \left \{ \frac{ \lambda_{\min}(Q_{i}) }{2 \beta_i},1 \right \}.
    \end{align}
     From (\ref{2dv_t_11}) and the comparison principle it follows that
    \begin{equation}\label{Vi_t}
        V_i(t) \leq e^{-\mu t}V_i(t_0),\quad t  \in [t_0, t_{1}),
    \end{equation}
    where 
    \begin{align}
         \mu &= \min \limits_{i=1,\ldots,p} \{\mu_i\}. \label{bita2}
        \end{align}
    Additionally, for every pair of modes \( i, j \), the following holds on the switching boundary \( \Omega_{i,j} = \Omega_{j,i}  \), that is, for  $X(t) \in \Omega_{i,j} $,
    \begin{align}\label{eq:exact_switching1}
        V_j(t) - V_i(t)  &=  X(t)^T (P_j - P_i) X(t) \notag \\ & \ + (b-b)\int_{t-D}^{t} e^{(\theta+D-t)}W(\theta)^2 d\theta \notag \\
                         & =   X(t)^T (P_j - P_i) X(t).
    \end{align}
    Applying (\ref{eq:exact_quadratic_switching}) in (\ref{eq:exact_switching1}) we get that, for  $X(t) \in \Omega_{i,j} $, 
    \begin{align}\label{eq:exact_switching2}
        V_j(t) - V_i(t)  &= 0  .
    \end{align}
    Finally, observing (\ref{LyapunovFnc2}), we have for all $t \in [t_0,t_1)$
        \begin{align}\label{kappas2}
        \kappa_{1}&\left(\left| X(t) \right|^2 + \int_{t-D}^{t} W(\theta)^2 d\theta \right)  \leq V_i(t) \notag \\
                &\leq  \kappa_{2}\left(\left| X(t) \right|^2 + \int_{t-D}^{t} W(\theta)^2 d\theta \right),
    \end{align}
    where 
    \begin{align}
        \kappa_1 &= \min \limits_{i=1,\ldots,p} \{ {\kappa_{1,i}}\}, \
        \kappa_2 = \max \limits_{i=1,\ldots,p}\{\kappa_{2,i}\}, \label{k22}
    \end{align}
    for
    \begin{align}
        \kappa_{1,i} &= \min \left\{\alpha_i, \frac{2|P_iB_i|^2}{\lambda_{\min}(Q_i)}\right\}, \,  \kappa_{2,i} = \max \left\{\beta_i, \frac{2|P_iB_i|^2}{\lambda_{\min}(Q_i)}e^D\right\}. 
    \end{align}
    Exponential stability for the system in the $(X, W)$ variables can now be proved. Let $s_1, s_2, \cdots, s_l$ be any ordered arbitrary switching times in $[0, t)$ with
$ 0 < s_1 < s_2 < \cdots < s_{l-1}< s_l < t $. Then according to (\ref{Vi_t}) and (\ref{eq:exact_switching2}), we have
    \begin{align}\label{Vsigma}
        V_{\sigma(X(t))}(t) &\leq e^{-\mu (t-s_l)}V_{\sigma(X(s_l))}(s_l) \notag \\
                            &\leq e^{-\mu (t-s_{l-1})}V_{\sigma(X(s_{l-1}))}(s_{l-1}) \ldots \notag \\
                            &\leq e^{-\mu t}V_{\sigma(X(0))}(0).
     \end{align}
    Setting $\kappa = {\kappa_2}/{\kappa_1}$, then from (\ref{kappas2}) and (\ref{Vsigma}) we reach (\ref{trans_stability_result2}).

\end{proof}  
    
\end{lemma}

\begin{customproof}[Proof of Theorem \ref{Exponential Stability}. ]
     Now we are able to complete the proof of Theorem~\ref{Exponential Stability}, and hence, conclude on the stability of the original system. Combining (\ref{ut_bound}), (\ref{wt_bound}), and (\ref{kappas2}), with the result in Lemma \ref{trans_stability2}, we get (\ref{stability equation}) where $\rho = \sqrt{\frac{2\kappa_1\nu_1\nu_2}{\kappa_2}}, \
    \xi  = \frac{\mu}{2}.$ 
\end{customproof}

\section{Simulation Example} \label{sec: sims}
We consider the switched system (\ref{1.1}) with two modes as
\begin{align}
    A_1 = \begin{bmatrix}2.5 & -1 \\
1.5 & 1.3\end{bmatrix},\;
    B_1 = \begin{bmatrix}1 \\ 0\end{bmatrix}, \; A_2 = \begin{bmatrix}0.1 & -3 \\
1.7 & 0.17\end{bmatrix},\;
    B_2 = \begin{bmatrix}0 \\ 1\end{bmatrix}.
\end{align}
To design feedback gains $K_i$, we follow the approach of multiple Lyapunov functions (as in~\cite[Sections~3.3 and 3.4]{Liberzon}). Hence, we solve the following set of 
LMIs using the ${S}$-procedure, for each mode $i=1,2$,
\begin{align}
&(A_1 P_1 + B_1 Y_1)^\top + (A_1 P_1 + B_1 Y_1) 
+ \gamma_1 (P_1 - P_2) \prec 0, \\
&(A_2 P_2 + B_2 Y_2)^\top + (A_2 P_2 + B_2 Y_2) 
+ \gamma_2 (P_2 - P_1) \prec 0, \\
&P_i \succ 0, \qquad i=1,2,
\end{align}
where $Y_i = K_i P_i$, which lead to $\gamma_1=\gamma_2=0.0001$, 
\begin{equation}\label{P1,K1}
    P_1 = \begin{bmatrix}2.4764 & -1.0101 \\
-1.0101 & 0.7256\end{bmatrix},\;
    K_1 = \begin{bmatrix} -5.6238 & -6.6582\end{bmatrix},
\end{equation}
\begin{equation}
    P_2 = \begin{bmatrix}1.3121 & 0.6749 \\
0.6749 & 1.3572\end{bmatrix},\;
    K_2 = \begin{bmatrix} 2.4007 & -2.2096\end{bmatrix}.
\end{equation}
We define the following quadratic forms 
\begin{equation}
   E_i(X) = X^\top P_i X, \qquad i=1,2,
\end{equation}
(which depend only on $X$), and the switching law is 
\[
\sigma\bigl(X(t)\bigr) = 
\max_i \{ \arg\max_{i \in \{1,2\}} E_i\bigl(X(t)\bigr)\},
\] which is motivated by applications in communication 
networks \cite{NetworkCom}, \cite{liu2015networked} and more specifically it is inspired by the Try-Once-Discard (TOD) protocol.  Hence, the $\Omega$ sets are
\begin{align}
    \Omega_1 &= \big\{ X \in \mathbb{R}^n ~\big|~ E_1(X) > E_2(X) \big\},  \\
    \Omega_2 &= \big\{ X \in \mathbb{R}^n ~\big|~ E_1(X) \leq E_2(X) \big\}, \\
        \Omega_{12} &= \Omega_{21} = \big\{ X \in \mathbb{R}^n ~\big|~ E_1(X) = E_2(X)\big\}. \label{Omega_ssurface}
\end{align}
With this switching signal, the matrices $K_i, \ P_i$ and the sets $\Omega_i$, defined in (\ref{P1,K1})--(\ref{Omega_ssurface}) satisfy Assumption~\ref{delay-free as}. Furthermore, we introduce a small hysteresis band in the switching rule, to practically avoid chattering behavior at the switching surfaces. We set $D=1$, initial conditions \(X(0)=[2;\,-1]\), $U(s)=0$, for $s \in [-D,0)$, and simulation time is \(T=10s\).

To compute the implicit predictor formula as in (\ref{P_theta}), we apply the left-end point rule iteratively over the delay interval $[t-D,\,t]$.  Denoting $h = \Delta t = 10^{-3}\,\mathrm{s}$ and $N = D / h$, we set for each $t$ with $\hat{P}(t-D) = X(t)$
\begin{equation}
  \begin{aligned}
    &\hat{P}\bigl(t - D + j\,h\bigr) 
    = \hat{P}\bigl(t - D + (j-1)\,h\bigr) \\
    & + {h}\Bigl[
        f\left(\hat{P}(t - D + (j-1)h\right),\,U(t - D + (j-1)h)\Bigr],
  \end{aligned}
\end{equation}
for $j=1, \ldots N$, where 
\[
  f\left(\hat{P},\,U\right)
  \;=\;
  A_{\sigma\left(\hat{P}\right)}\,\hat{P}
  + B_{\sigma\left(\hat{P}\right)}\,U\,.
\]
The implicit predictor 
formula \eqref{P_theta} instead of the semi-explicit expression \eqref{P(t)} is used because here this choice 
is computationally faster. Nevertheless, the accuracy and computational complexity of the implicit and the semi-explicit formulas may depend on the size and structure of the system, the form of the $\Omega_i$ sets, as well as the specific approach employed for actual implementation of the two schemes. Hence, a more detailed study would be needed for providing concrete recommendations on which computational method would be the most efficient.

Fig.~\ref{fig2} illustrates  the evolution of the actual switching signal \(\sigma(X(t))\), with blue segments indicating the active mode and vertical black lines marking the switching instants. Fig.~\ref{fig3} shows the state trajectories and the performance of the controller. Fig.~\ref{fig4} depicts the phase portrait 
of the system trajectories. The switching regions $\Omega_1$ and $\Omega_2$ are depicted in different colors, with the hysteresis band shown in gray.

\begin{figure}[ht!]
    \centering
    \includegraphics[width=8.3 cm]{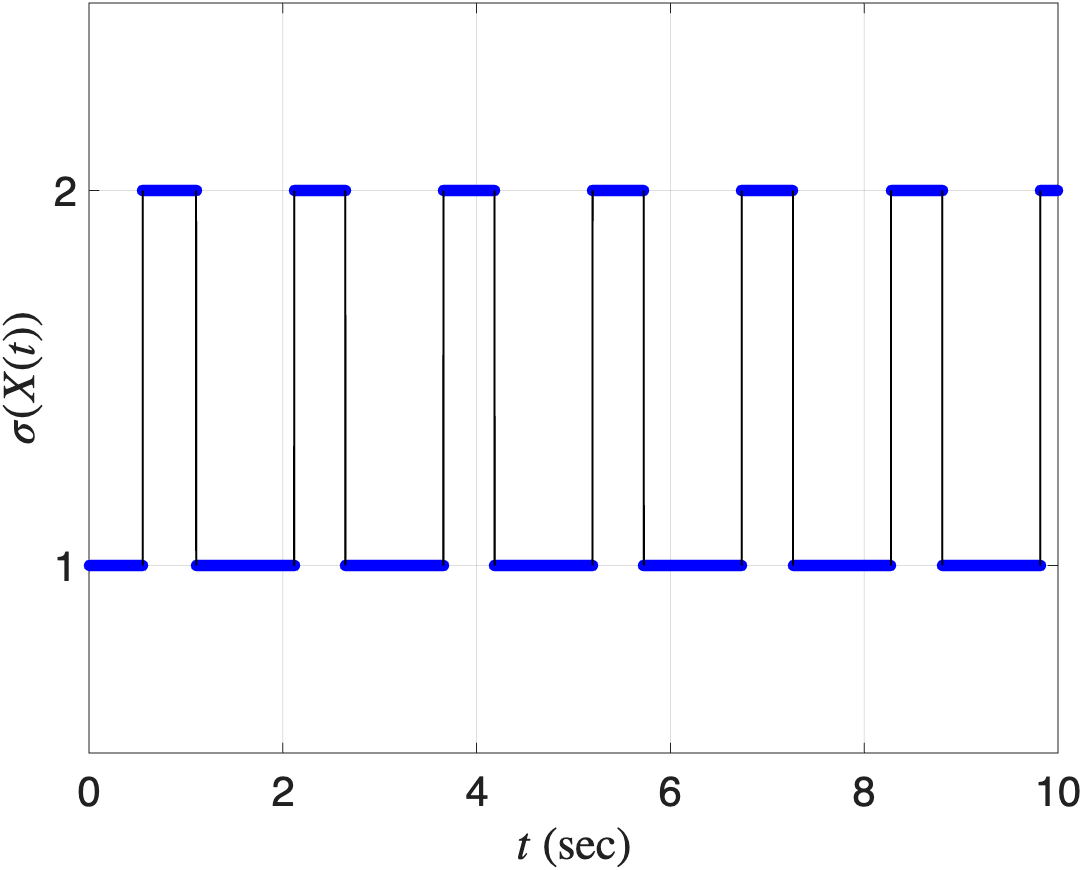}
    \caption{Evolution of the switching signal \(\sigma(X(t))\).
    \label{fig2}}
\end{figure}

\begin{figure}[ht!]
    \centering
    \includegraphics[width=8.3 cm]{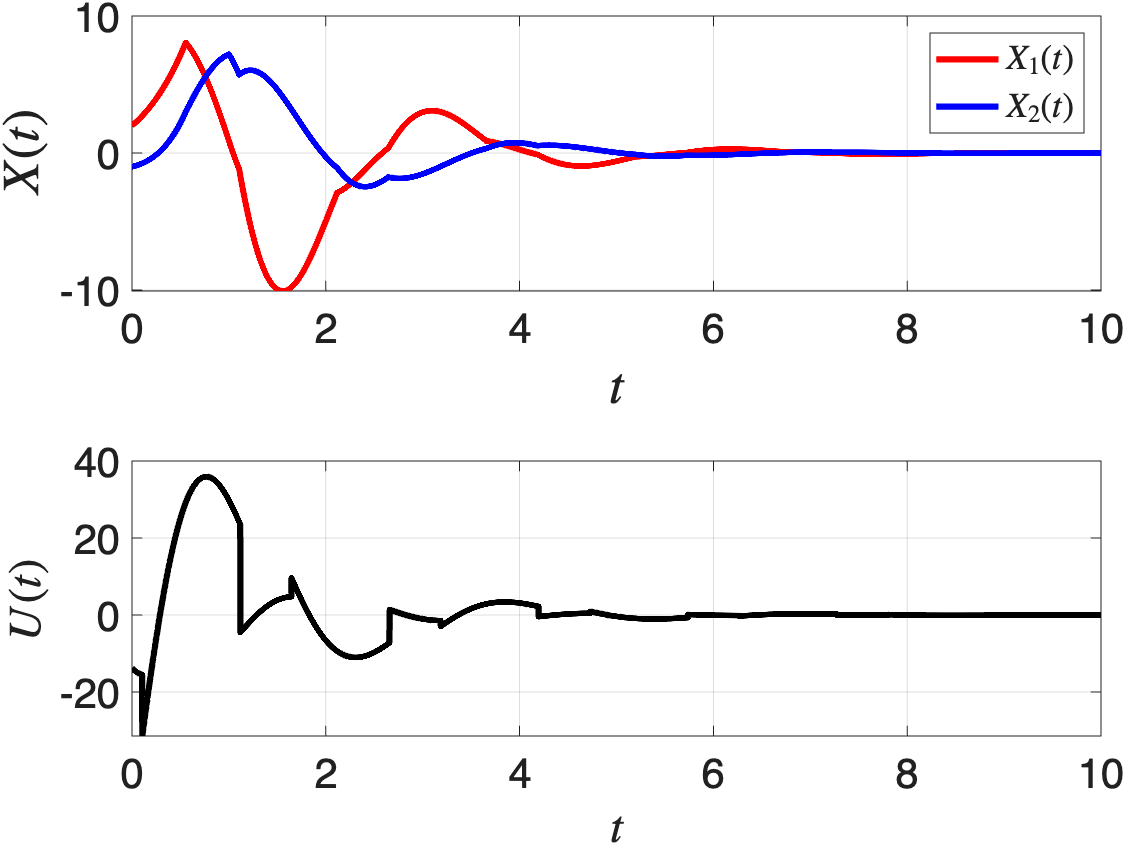}
    \caption{Evolution of state $X(t)$ and control input $U(t)$. 
    \label{fig3}}
\end{figure}

\begin{figure}[ht!]
    \centering
    \includegraphics[width=8.3 cm]{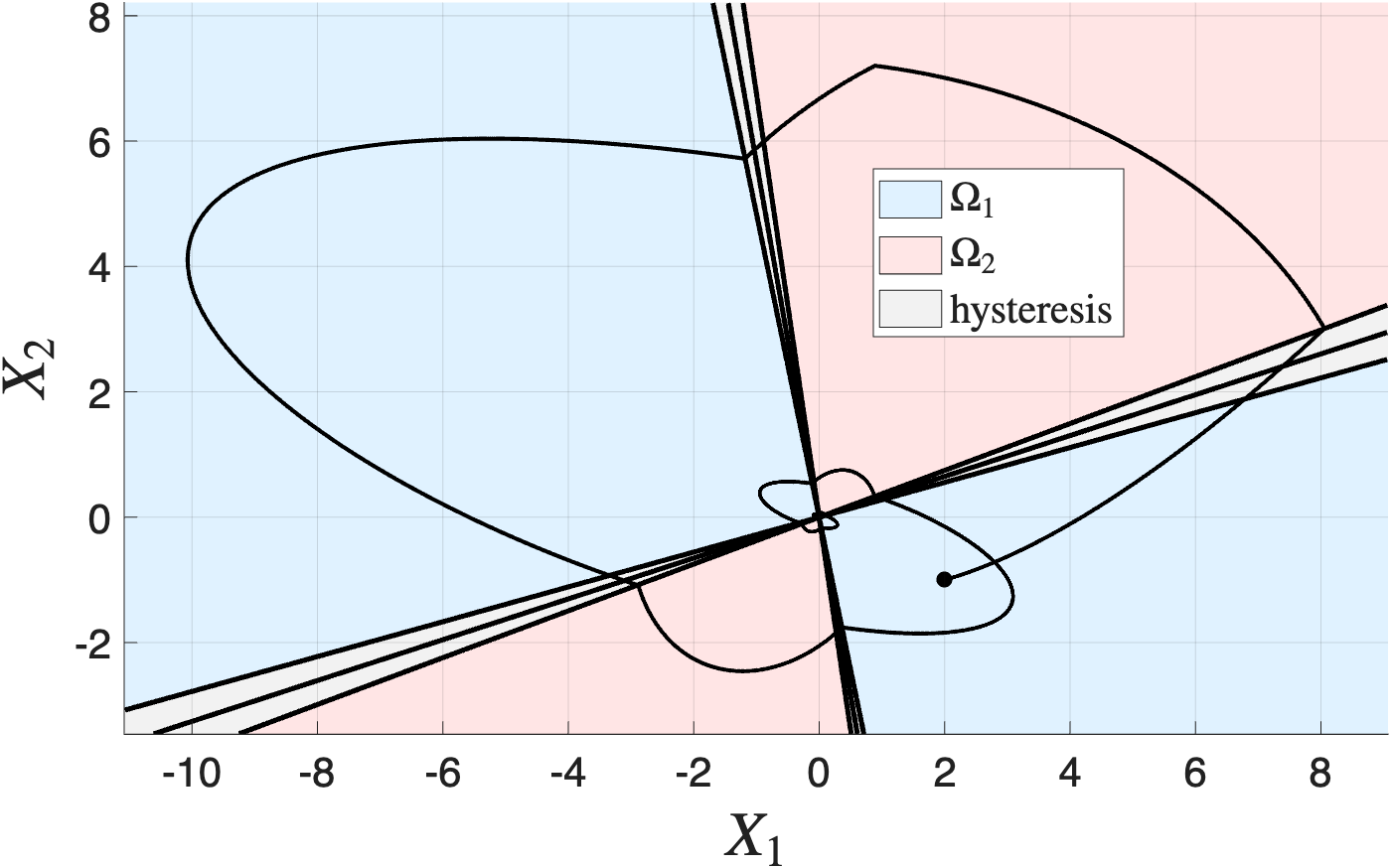}
    \caption{Phase portrait of the system trajectories along with the switching regions $\Omega_1$ (blue), $\Omega_2$ (red), and hysteresis band (gray).}
    \label{fig4}
\end{figure}

\section{Conclusions and Future Works}\label{sec5} 

In this work we developed a predictor-feedback control law for a class of switched linear systems subject to state-dependent switching and arbitrarily long, constant input delay. The main element of our design is an exact predictor state that we constructed and which can be computed in an implicit manner, or using a semi-explicit formula that we provided. We established uniform exponential stability of the closed-loop system constructing multiple Lyapunov functionals via utilization of backstepping. A two-mode example demonstrated the effectiveness of the proposed controller.

As next step we intend to develop a predictor-feedback control design methodology for a larger class of systems with state-dependent switching, adopting a hybrid systems framework for relaxing the necessary restrictions imposed here by the specific form of the switching signal and system considered. We also intend to investigate the computational complexity and accuracy of both the implicit and the semi-explicit implementation of our predictor-feedback control law, towards providing concrete guidelines for its efficient implementation.

\bibliography{IEEEex}      
\bibliographystyle{IEEEtran}

\end{document}